\documentclass[runningheads,envcountsame]{llncs}
\usepackage{graphicx}
\usepackage[usenames,dvipsnames,svgnames,table]{xcolor}
\usepackage[colorlinks=true,linkcolor=blue,citecolor=ForestGreen]{hyperref}
\usepackage{amsmath,amssymb,enumerate}
\usepackage{algorithm}
\usepackage{algorithmic}
 \usepackage{orcidlink}
\usepackage{scrextend}
\usepackage[utf8]{inputenc}
\usepackage[T1]{fontenc}

\usepackage{paralist}
 
\usepackage{bm}
\usepackage{hyperref}
\usepackage{booktabs}
\usepackage[capitalize]{cleveref}
\usepackage{float}
\usepackage[color=green!40,textsize=small]{todonotes}

\newcommand{\trw}{\mathsf{tw}} 
\newcommand{\pw}{\mathsf{pw}} 
\newcommand{\s}{\mathsf{s}} 
 
\newcommand{\cop}{\mathsf{cop}}
 \newcommand{\apx}{\mathsf{apxcop}}
\newcommand{\capt}{\mathsf{capt}}

\Crefname{figure}{Figure}{Figures}
\Crefname{claim}{Claim}{Claims}

\crefformat{equation}{\textup{#2(#1)#3}}
\crefrangeformat{equation}{\textup{#3(#1)#4--#5(#2)#6}}
\crefmultiformat{equation}{\textup{#2(#1)#3}}{ and \textup{#2(#1)#3}}
{, \textup{#2(#1)#3}}{, and \textup{#2(#1)#3}}
\crefrangemultiformat{equation}{\textup{#3(#1)#4--#5(#2)#6}}%
{ and \textup{#3(#1)#4--#5(#2)#6}}{, \textup{#3(#1)#4--#5(#2)#6}}{, and \textup{#3(#1)#4--#5(#2)#6}}

\Crefformat{equation}{Equation \textup{#2(#1)#3}}
\Crefrangeformat{equation}{Equations \textup{#3(#1)#4--#5(#2)#6}}
\Crefmultiformat{equation}{Equations \textup{#2(#1)#3}}{ and \textup{#2(#1)#3}}
{, \textup{#2(#1)#3}}{, and \textup{#2(#1)#3}}
\Crefrangemultiformat{equation}{Equations \textup{#3(#1)#4--#5(#2)#6}}%
{ and \textup{#3(#1)#4--#5(#2)#6}}{, \textup{#3(#1)#4--#5(#2)#6}}{, and \textup{#3(#1)#4--#5(#2)#6}}

\crefname{section}{section}{sections}

\usepackage{tikz}
\tikzset{normalnode/.style={circle, draw, fill=black, inner sep=0, minimum width=1.5mm}}

\usetikzlibrary{fit}
\usetikzlibrary{shapes,arrows,shadows,positioning}
\usetikzlibrary{backgrounds}

\renewcommand{\le}{\leqslant}

\renewcommand{\leq}{\leqslant}
\renewcommand{\geq}{\geqslant}

\newcommand{\comment}[1]{}

\begin{document}

\title{Capturing an Invisible Robber using Separators
}
%\thanks{PhD Student’s Paper}

%\date{}

  \author{Igor Potapov \orcidlink{0000-0002-7192-7853} \and
Tymofii Prokopenko \orcidlink{0009-0002-1922-7718}
\and
John Sylvester \orcidlink{0000-0002-6543-2934}
}

\institute{Department of Computer Science, University of Liverpool, UK
\email{\{potapov,t.prokopenko,john.sylvester\}@liverpool.ac.uk}\\
}

%\begin{document}

\maketitle
%%%%%%%%%%%%%%%%%%%%%%%%%%%%%%%%%%%%%%%%%%%%%%%%%%
\begin{abstract}
    We study the zero-visibility cops and robbers game, where the robber is invisible to the cops until they are caught. This differs from the classic game where full information about the robber's location is known at any time. A previously known solution for capturing a robber in the zero-visibility case is based on the pathwidth decomposition. We provide an alternative solution based on a separation hierarchy, improving capture time and space complexity without asymptotically  increasing the zero-visibility cop number in most cases. In addition, we provide a better bound on the \textit{approximate} zero-visibility cop number for various classes of graphs, where approximate refers to the restriction to polynomial time computable strategies.

\end{abstract}
%\codes{05C57,68R10}

%%%%%%%%%%%%%%%%%%%%%%%%%----Introduction----%%%%%%%%%%%%%%%%%%%%%%%%%
\section{Introduction}\label{sec:intro}

Pursuit-evasion games model scenarios in which a group of pursuers attempts to capture an evader within a defined `arena'. One of the most extensively studied pursuit-evasion games played on graphs is the \textit{Cops and Robbers} game. This two-player game was introduced in \cite{nowakowski1983}  and \cite{quilliot1978jeux}, as a perfect-information game on a finite simple graph \(G\). 
To begin, the \(k\) cops choose starting vertices; after seeing them, the robber chooses a vertex. Then in each round, any cop may stay or move to a neighbour and, then the robber moves likewise. Multiple cops may share a vertex, and all positions are visible to the players. The cops win if some cop ever occupies the robber’s vertex; otherwise the robber evades forever.
 
The original formulation actually considered a single cop ($k=1$) chasing a single robber and leads to a characterisation of graphs where the cop can always guarantee a capture, known as cop-win graphs. Since then, the game has been the subject of significant research, including generalizations involving multiple cops, robbers, and variations in rules. For a comprehensive overview of the topic, we refer the reader to the book by Bonato \& Nowakowski \cite{bonato2011game}.
There are many variants of the %classic
Cops and Robbers game, distinguished by the restrictions placed on the cops and the robber.  It is natural to restrict the speed or edges that can be used by the cops or robber \cite{frieze2012variations,EnrightMPS23}. 
Another important variant restricts the cops’ information about the robber’s position during the game, leading to the concept of \textit{\(\ell\)-visibility Cops and Robbers} \cite{clarke2020limited}, where the cops can detect the robber only within a limited distance.

In this paper, we focus on the \textbf{\textit{zero-visibility Cops and Robbers}} variant, in which the cops have no information about the robber’s position at any point during the game. This version was first introduced in \cite{Tosic} and later investigated further in \cite{dereniowski2013,dereniowski2015zero}. 
It can be viewed as a special case of the 
 \textit{\(\ell\)-visibility Cops and Robbers} with parameter \(\ell\) equal to $0$. A characterisation for the \textit{\(\ell\)-visibility cop number} of a tree based on its structure was given in \cite{clarke2020limited} and shows that the difference between the 
$\ell$-visibility cop number and the %classic
cop number can be arbitrarily large. 
  Moreover, both variants of the game are closely related to the \textit{limited visibility graph search problem}, in which the cops move at speed one and the robber has infinite speed. This problem was recently introduced in \cite{kehagias2021algorithm}, providing an algorithm for capturing the robber under these conditions.

The problem of estimating the \emph{zero-visibility cop number}, i.e.\ the minimum number of cops to win the zero-visibility Cops and Robbers game, was initiated in \cite{tovsic1985vertex}. Tang \cite{tang2005cops} gives lower bounds for connected graphs, showing that the zero-visibility cop number exceeds half of the minimum vertex degree of the graph, also establishing the connection to the pathwidth, and characterisation of zero-visibility cop-win graphs. Moreover, \cite{tang2005cops} constructs a graph with an arbitrarily large \textit{zero-visibility cop number} and capture the robber using an optimal path decomposition. Later, a class of graphs on which two cops are sufficient in the \textit{zero-visibility Cops and Robbers} game was provided in \cite{jeliazkova2006aspects}.

In \cite{dereniowski2015zero}, the authors studied monotonic strategies to capture the robber, making a connection between the \textit{zero-visibility cop number}, its monotone variant, and the pathwidth of a graph. They proposed an algorithm for trees, and exhibited graphs of zero-visibility cop number two with arbitrarily large pathwidth.
The same authors, later in \cite{dereniowski2015complexity},  addressed the computational complexity of the \textit{zero-visibility Cops and Robbers} game, presenting a linear-time algorithm to compute the \textit{zero-visibility cop number} of a tree. They also showed that the corresponding decision problem is NP-complete on a non-trivial class of graphs.
 
In \cite{xue2019partition}, the authors developed a partition method to establish lower bounds on the \textit{zero-visibility cop number} in graph products. Subsequently, they analysed the \textit{zero-visibility Cops and Robbers} on graph joins, lexicographic products, complete multipartite and split graphs. In these studies, they refined the lower bounds on the \textit{zero-visibility cop number} for these classes of graphs compared to \cite{tang2005cops} and established a connection to the graph matching number.
In \cite{yang2022one} the authors provide an algorithm for computing optimal cop-win strategies on a tree in the \textit{one-visibility Cops and Robber} game and show the relationship between the \textit{one-visibility cop number} and \textit{zero-visibility cop number} of trees. 

For arbitrary graphs, the best-known strategy for bounding the \textit{zero-visibility cop number} is based on the initial idea of pathwidth decomposition \cite{tang2005cops}. However, the  weaknesses of the pathwidth solution in respect to capture time, computational complexity and quality of the approximation has not been addressed in the literature and motivated our study for alternative solutions, for example based on separators that have been used in full visibility case, see \cite{JoretKT10,LohO17}.

In this paper, we design an alternative algorithm based on a separation hierarchy that improves the capture time for a robber without an increase to the zero-visibility cop number in most cases. More specifically, our algorithm provides \( O\big(D \cdot \frac{n}{f(n)}\big) \)  capture time compared to \( O\left(D \cdot n\right) \) of the \texttt{PW} algorithm in any hereditary \(f\)-separable class of graphs with diameter \(D\).
Also, since the optimal pathwidth decomposition is NP-complete, with the alternative separation hierarchy approach we can improve the approximate zero-visibility cop number required to win in several classes of graphs.
Our approach provides a wider applicability and improvement facilitated by the existence of better approximation algorithms for computing the separation sets of graphs.
Moreover, the approach allows us to cover classes of graphs where an optimal pathwidth decomposition, or a good approximation, has not been established yet.
 
As an example, we also demonstrate the applicability of the algorithm to the hyperbolic random graph.
Finally, the space complexity of pathwidth decomposition has only a trivial \( O(n^2) \) upper bound, and with the alternative approach we show that $O(n)$ space is enough. 

 Proofs of all claims can be found in the Appendix.

%%%%%%%%%%%%%%%%%%%%%%%%%----Preliminaries----%%%%%%%%%%%%%%%%%%%%%%%%%
\section{Preliminaries}\label{sec:prelm}

%\noindent
%\textit{Graphs.}
 \paragraph{Graphs} 
  We consider finite undirected graphs $G=(V,E)$, without loops or multiple edges, and let $V(G)$, or just $V$ if clear from context, be the vertex set of $G$, and $E(G)$ or $E$ be the edge set. 
%   We call a tree $T$ any finite acyclic graph
  The diameter $D_G$ of a graph $G = (V,E)$ is the maximum distance between any pair of vertices, % $D_G := \max_{u,v \in V}\left(\text{dist}(u, v)\right)$, 
  where, the distance between two vertices is the minimum number of edges on any $u$-$v$ path.  
 Given a graph $G$ and a vertex subset $S\subseteq V$, we let $G[S]$ be the graph induced by $S$, i.e.~the graph with vertex set $S$ and edge set $\{uv \in E: u,v \in S\}$.
  A class of graphs is \emph{hereditary} if it is closed under taking induced subgraphs, i.e.~for any hereditary class $\mathcal{C}$ of graphs, if $G\in \mathcal{C}$ then $H \in \mathcal{C}$ for all induced subgraphs $H$ of $G$.

A \textit{path decomposition} of a graph $G=(V,E)$ is a sequence $B = (B_1, \dots, B_r)$ of subsets of $V$ (or `bags') if the following conditions are satisfied:
1) $\bigcup_{i =1}^{r} B_i = V$
2) for every edge $vw \in E$ there exists $i \in {1,...,r}$ with $\{v,w\} \subseteq B_i$
    %\item For every edge $e$ of $G$, some $X_i(1 \leqslant i \leqslant r)$ contains both ends of $e$.
3) For $1 \leqslant i \leqslant j \leqslant k \leqslant r, B_i \cap B_{k} \subseteq B_{j}$.
%\end{itemize}
%
In other words, in the path decomposition, every vertex appears in at least one bag; for every edge, both endpoints appear together in some bag and for any vertex, the bags containing it appear consecutively along the path.
The \textit{width} $\pw(B)$ of $B$ is defined as $\max_{i \in [r]} |B_i| -1$. 
The \textit{path\-width} $\pw(G)$  of $G$ is the minimum width of any path decomposition of $G$.

% \paragraph{Separation.} 
Let $\tfrac12 \leq \alpha < 1$ be a real number, $s\geq 0$ an integer, and $G=(V,E)$ a graph. A subset $S\subset V$ is said to be an $(s,\alpha)$-\textit{separator} of $G$, if there exist subsets $A,B \subset V$ such that
%\begin{itemize}
	%\item 
    1) $V = A \cup B \cup S$ and $A, B, S$ are pairwise disjoint;
	%\item 
    2) $|S| \leq s$, $|A|, |B| \leq \alpha |V|$; and
	%\item 
    3)$|\{ab\in E : a\in A, b\in B \}| = 0$.
%\end{itemize}
We will refer to $A,B$ as \emph{separated sets} and $S$ as the \emph{separator}. 
The \textit{separation number} $\s(G)$ of $G$ is the smallest $s$ such that all subgraphs $G'$ of $G$ have an $(s,2/3)$-separator. For a function $f:\mathbb{N}\rightarrow \mathbb{N}$, we say that a class $\mathcal{C}$ is \emph{$f$-separable} if any $n$-vertex graph $G\in \mathcal{C}$ satisfies $\s(G) \leq f(n)$.  

\begin{figure}[t]
    \centering
    \includegraphics[width=.9\textwidth]{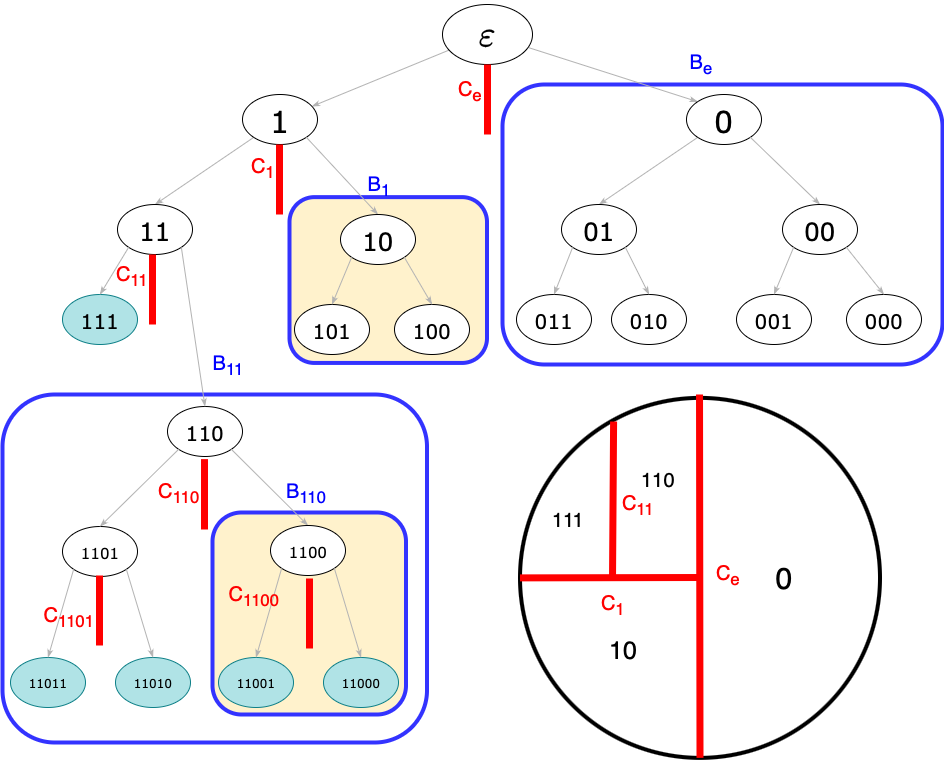}
    \caption{
    The \texttt{STT} algorithm: Separation Tree, Stack Updates, and Graph Partitioning.}
    \label{fig:example}
\end{figure}
By a \textit{separation algorithm}, we refer to an algorithm \(\mathcal{S}\) that takes a graph \(G = (V, E)\) as input and returns subsets \(A\), \(B\), and \(S\) of \(V\), such that \(S\) is an \((s, \alpha)\)-separator of \(G\). 
We say that an \emph{$f$-separable} class \(\mathcal{C}\) has an \textit{\((f, T(n))\) - separation algorithm} if there exists a separation algorithm \(\mathcal{S}\) such that for any \(n\)-vertex graph \(G \in \mathcal{C}\), the algorithm returns an \((f(n), 2/3)\)-separator of \(G\) in time at most \(T(n)\), where \(f(n)\) and \(T(n)\) are functions from \(\mathbb{N} \to \mathbb{N}\).
As an example, in our notation, the celebrated result of Lipton-Tarjan \cite{LiptonT80} gives an \textit{$(O(\sqrt{n}),O(n))$-separation algorithm} for planar graphs. 
%
 % 

%\paragraph{Binary separation tree.}
Let \(G\) be an \(n\)-vertex graph and \(g(n)\geq 0\) be a function that defines a threshold based on \(n\). 
The \textit{binary separation tree} $\mathcal{T}_{\mathsf{sep}}$
of a graph $G$, with respect to a separation algorithm $\mathcal{S}$ and a threshold function $g(n)$, is defined as follows: the nodes of $\mathcal{T}_{\mathsf{sep}}$
are subsets of $V(G)$, derived by starting with the root node $\varepsilon$ corresponding to the subset $V(G)$. 
For each node \( \omega \in \mathcal{T}_{\mathsf{sep}} \), if \( |\omega| > g(n) \), we apply \( \mathcal{S} \) to \( G[\omega] \) to obtain subsets \( \omega \cdot``1" \), \( \omega \cdot ``0" \), and \( \operatorname{Sep}(\omega) \). The subsets \( \omega \cdot``1" \) and \( \omega \cdot``0" \) are added as nodes in $\mathcal{T}_{\mathsf{sep}}$ if they are not already present. The children of \( \omega \in \mathcal{T}_{\mathsf{sep}} \) are defined as follows: \( \omega \) has children \( \omega \cdot``1" \) and \( \omega \cdot``0" \) if \( |\omega| > g(n) \) and the two are obtained by applying \(\mathcal{S}\) to \( G[\omega] \). A node \( \omega \) is a leaf in $\mathcal{T}_{\mathsf{sep}}$ if and only if \( |\omega| \leq g(n) \), indicating that \( G[\omega] \) cannot be separated further by \( \mathcal{S} \). See \cref{fig:example} for an example of a binary separation tree.
%
%
%It is important to mention 
This is similar to the notion of \textit{separator hierarchy}, when the iterative separation is continued until the separated sets are of size one, see~\cite{blasius2016hyperbolic}.

\paragraph{Cops and Robbers}

The original cops and robber game \cite{AIGNER19841} is played on an undirected $n$-vertex graph $G = (V, E)$ with two players: one controls a set of $k \geq 1$ cops, indexed by the integers $0, \ldots, k-1$, and the other controls a single robber. We will generally refer to the moves of the cops and robbers rather than to the players themselves. In round 1, the cops choose their initial positions on the graph, occupying a set of vertices $c_0(0), \ldots, c_{k-1}(0)$. The robber then selects an unoccupied vertex $r(0)$ to begin. In each subsequent round $t = 2, 3, \ldots$, the cops move first, with each cop moving from their current vertex $c_i(t-1)$ to a vertex $c_i(t) \in N^{+}(c_i(t-1))$, where $N^{+}(v) := N(v) \cup \{v\}$ is the closed neighbourhood of vertex $v$ and $N(v) := \{ u \in V \mid(v, u) \in E \}$ is the open neighbourhood of vertex $v$.
The robber then moves in a similar way, choosing a new position $ r(t) \in N^{+}(r(t-1))$. If a cop or the robber chooses to remain on their current vertex, this is referred to as a pass. We will refer to this variant of the game as the \textit{classic Cops and Robbers} game. In the version where cops are allowed to move to any vertex $c_i(t) \in V$ of the graph $G$, the game is called \emph{node search} (or \emph{teleporting Cops and Robbers} game) \cite{WAGNER2015107,peng2000edge} i.e.~the cops are not restricted to moving only to neighbouring vertices, while the robber continues to move along the edges of the given graph as before.

The original game is played with perfect information: both the cops and the robber know each other's positions at all times. The cops win if, in some round, one of them moves to the vertex occupied by the robber, thereby capturing the robber. In contrast, the robber wins if they can avoid capture indefinitely. Both the cops and the robber are assumed to play optimally: the cops aim to minimise the number of rounds until capture, while the robber seeks to maximise it. The cop number of a graph $G$, denoted by $\cop(G)$, is the minimum number of cops required to guarantee the capture of the robber in $G$.

A graph $G$ is said to be a $k$ cop-win graph if it admits a cop strategy that guarantees capture using at most $k$ cops. 
The \textit{capture time  of a strategy} is defined as the maximum number of rounds until capture is achieved \cite{BrandtEUW17}. 
The \textit{capture time of  a graph $G$}, $\capt(G,k)$, is then defined as the minimum capture time of any cop strategy on $G$ with $k$ cops. The \textit{capture time of a graph} was first considered by Bonato, Golovach, Hahn \&  Kratochv\'il \cite{BonatoGHK09}, who mostly considered the \textit{capture time of cop-win graphs}. 
The \textit{capture time of planar graphs} using three cops was shown to be at most $2n$ in \cite{PisantechakoolT16}. One can show that if the \textit{capture time} by $k$ cops is finite, then it is at most $O(n^{k+1})$; surprisingly, this bound is tight \cite{BrandtEUW17}.
We study a variant of the classic Cops and Robber game on a graph \( G = (V, E) \), known as the \textit{zero-visibility Cops and Robbers} game, first introduced in \cite{Tosic} and further explored in \cite{dereniowski2013,dereniowski2015zero,tang2005cops} where the bounds on the zero-visibility cop number, denoted by $\cop_0(G)$, in terms of path width were established, and an algorithm was developed for capturing a robber. 
In this version, the robber is invisible to the cops, who can only detect and capture the robber by occupying the same vertex. The cops move without knowing the location of the robber, while the robber, who sees the cops, moves optimally to evade capture. Similarly, we let $\capt_0(G,k)$ be the capture time by $k$ cops in the zero-visibility setting. 
\paragraph{A Cops and Robbers Path Width {\upshape(\texttt{PW})} algorithm}
In \cite{tang2005cops}, a strategy to capture the robber using an optimal path decomposition has been proposed, which we refer to further as the \texttt{PW} algorithm. 
The minor modification of the algorithm was used in \cite{dereniowski2015zero} to study the monotonic strategies 
and investigate the connection between the monotone zero-visibility cop number, the zero-visibility cop number, and the pathwidth of a graph.

The \texttt{PW} algorithm is based on finding a path decomposition \( B = \{B_1, \dots, B_r\} \) of a connected graph \( G \) such that \( \pw(B) = \pw(G) \), with the additional condition that $\forall i \in \{1, \dots, r-1\}$, the sets \( B_i \setminus B_{i+1} \), \( B_{i+1} \setminus B_i \), and \( B_i \cap B_{i+1} \) are not empty. In the algorithm, \( k = \pw(G) \) cops are used, which move from the bag \( B_i \) to \( B_{i+1} \) for each \( i \in \{1, \dots, r-1\} \) in such a way that the robber's territory \( R_i \) monotonically decreases after each round \( i \). 
Initially, the \( k \) cops are placed on the vertices of \( B_1 \). If \( k > |B_1| \), then multiple cops can be placed on the same vertex of \( B_1 \). The algorithm then proceeds by sequentially moving the cops from \( B_i \) to \( B_{i+1} \), so that they are moved from \( B_i \cap B_{i+1} \) to \( B_{i+1} \setminus B_i \), ensuring that the robber cannot enter \( B_i \) for any \( i \in \{1, \dots, r\} \). The movement in each step can be described as follows: for each \( i \in \{1, \dots, r-1\} \), the cops are moved from \( B_i \cap B_{i+1} \) to \( B_{i+1} \setminus B_i \), while the robber cannot move into \( B_i \). 
The algorithm halts when the robber is caught in any round \( i \). Since the robber’s territory \( R_i \) monotonically decreases with each round \( i \), in the worst case, the algorithm will need to pass through all \( r \) bags of the decomposition. 

Neither \cite{dereniowski2015zero} nor \cite{tang2005cops} consider using this algorithm to bound the capture time, however a simple analysis of the algorithm gives the following result. 
%shown below.

\begin{proposition}\label{prop:pw_capt_time_alg}
  The capture time of the {\upshape\texttt{PW}} algorithm with \(\pw(G)\) cops for a graph \( G \) with \( n \) vertices and diameter \( D_G \), is bounded above by \( O(D_G \cdot n) \) and the space complexity by $O(n^2)$.
\end{proposition}
\begin{proof}
Following the \texttt{PW} algorithm in \cite{dereniowski2015zero,tang2005cops} 
    we move the cop team from one bag \(B_i\) to the next bag \(B_{i+1}\) following a path decomposition \(B = (B_1, \dots, B_r)\) for all \(i \in {1,...,r-1}\). The number of bags in the path decomposition is bounded above by \(O(n)\), and to move a cop from \(B_i\) to \(B_{i+1}\), we need \( D_G \) moves.
    Since the size of each bag    can be as large as $O(n)$, and the total number of bags is bounded above by $n$, we need to store all bags to follow the \texttt{PW} algorithm. So, the \texttt{PW} algorithm requires $O(n^2)$ memory. \qed

\end{proof}

\section{Capturing invisible robber using separation hierarchy}
We present an algorithm for capturing a robber based on the traversal of the separation tree which is an alternative to the existing algorithm based on a pathwidth decomposition \cite{dereniowski2015zero}. 
During the design of the overall strategy, we first reason with respect to the node search (or teleporting Cops and Robbers) game \cite{WAGNER2015107}. In the node search model, cops are allowed to be placed on or removed from any vertex (i.e.~not restricted to adjacent moves), while the robber still moves along edges. Then, when we analyse the capture time and other aspects, we add back in the travel times of the cops (restricted to travel over edges).

The core idea of the algorithm is to iteratively apply a separation algorithm $\mathcal{S}$ that partitions the graph $G$ into three parts: a separator of size at most $f(n)$, and two subgraphs with no edges between them.  The separator is guarded by $f(n)$ cops, and two separate sets can be cleaned independently and recursively. Our ``Separation Tree Traversal'' (STT) algorithm operates on a separation hierarchy of the graph, but it does not require building the full separation tree in advance. Instead, the algorithm incrementally traverses this implicit hierarchy, which can be represented by tree structure, cleaning one subgraph at a time. At each step, it uses separators to isolate ``uncleaned'' regions, placing cops as static guards on the separator vertices to prevent the robber from moving into the already cleaned parts of the graph. If at any point an isolated uncleaned subgraph $G'$ of size at most $f(n)$ remains, it can be cleaned directly by deploying $f(n)$ cops to cover all its vertices simultaneously. 
This process is monotonic, so once a region of the graph is cleaned, it remains secure, and the robber cannot re-enter it. By maintaining this invariant controlling the cops' movements
and use of separators, the algorithm ensures that the entire graph is eventually cleaned and the robber is captured. We demonstrate the correctness of the algorithm and evaluate the capture time for a robber on a graph as well its time/space complexity by analysing a recursive function that estimates the number of vertices in the binary separation tree of the graph.
 
\subsection{Separation Tree Traversal (STT) algorithm}
 
The Separation Tree Traversal Algorithm (\texttt{STT}), which is given in \cref{alg:STT}, takes as input an $n$-vertex graph $G \in \mathcal{C}$, where $\mathcal{C}$ is an $f$-separable hereditary class, and defined with respect to a separation algorithm $\mathcal{S}$ and had a related separation tree $\mathcal{T}_{\mathsf{sep}}$.
In $G$, an unknown vertex $u$ contains a robber. It gives a sequence of cop moves, moving $O(f(n))$ cops during an iteration, which guarantee the capture of an invisible robber.

The \texttt{STT} algorithm consists of three functions: \texttt{Scheduler}, \texttt{Separate}, and \texttt{Clear}.
The main component of the \texttt{STT} algorithm is the \texttt{Scheduler}, which calls the \texttt{Separate} and \texttt{Clear} methods during its execution. \texttt{Scheduler} uses two stacks: {\sc Stack-B} which stores separated sets; and {\sc Stack-C} containing  separators. The algorithm terminates after monotonically after clearing all vertices.  %See Appendix~\ref{Appendix-A} for pseudocode.
%\cref{alg:sep} 

The \texttt{Clear}$(H)$ function simply moves a cop to each vertex of the input graph $H$ for a round then removes them, catching any robber that is confined to $H$. 

The input to \texttt{Separate} is an induced subgraph $H\subseteq G$, the index $V(H)$, and an integer $n$ (the size of the original graph $G$). If $|V(H)|\leq f(n)$, then \texttt{Clear} is called to catch any potential robber in $H$.  Otherwise $\mathcal{S}(H)$ is called, that is, the separation algorithm partitions $V(H)$ into three parts $A, B, C$. The separated subsets $A,B$, with indices extending that of $V(H)$ by a $0$ or $1$, are placed on {\sc Stack-B}, and the separator $C$, with the index of $V(H)$, is placed on {\sc Stack-C}.

The \texttt{Scheduler} uses two stacks, {\sc Stack-B} and {\sc Stack-C} where it stores separated sets and separators obtained by the separation algorithm $\mathcal{S}$, respectively. Each of these sets is indexed by its place in the separation tree, stored is a binary string. {\sc Stack-B} stores sets of separated vertices, and their indices, which have yet to be cleared.  
The subsets of vertices contained in {\sc Stack-B} are disjoint. Similarly, the disjoint separators, corresponding to sets in {\sc Stack-B}, are stored together with their indices in {\sc Stack-C}. The idea is that cops are temporarily placed on vertices of sets in $C$ to restrict movement of the robber between the two corresponding separated sets (from {\sc Stack-B}) so that these two separated sets can be cleaned independently. The role of the scheduler is to keep track (via the indices and stacks) of which sets are cleaned and assign cops accordingly.

The \texttt{Scheduler} starts by initially setting {\sc Stack-C} to be empty, and {\sc Stack-B} to contain the set $V(G)$ with index $\varepsilon$ (the empty string). The outermost while loop of the \texttt{Scheduler} runs until the  {\sc Stack-B} is empty, i.e.\ when there are no `uncleared' vertices. If {\sc Stack-B} is not empty, then the top element $U$ is popped. Its next step is to check if cops can be pulled off some of the separators. It does this by seeing if the separator at the top of {\sc Stack-C} separates the $U$ from its sibling in the separation tree, as if not then both of its corresponding separated sets must already be clear -  so the separator can be popped and its cops removed. The   \texttt{Scheduler} then calls \texttt{Separate} on $G[U]$ and starts the loop again.

%%%%%%%%%%%%%%%%% General Algorithm %%%%%%%%%%%%%%%%%%%%%%%

Given a string $w = w_1w_2\cdots w_n$ the function $\operatorname{DelLast}(w)$ returns the string $w_1w_2\cdots w_{n-1}$, i.e.\ it simply removes the last character from $w$. 
%Pseudocode of the Separation Tree Traversal (STT) algorithm:\\
%%%%%%%%%%%%%%%%%%%%%%%%%%%%%%%%%%%%%%%%%%%%%%%
%%%%%%%%%%%%%%%%%%%%%%%%%%%%%%%%%%%%%%%%%%%%%%%
\vspace{-2em}
\begin{algorithm}
\small
\footnotesize 
    \caption{\sl{Separation Tree Traversal} of $G$ with Separation Algorithm $\mathcal{S}$}\label{alg:STT}
    \begin{algorithmic}    
    \vspace{0.5em}
    \STATE {\textbf{Function} \texttt{Scheduler(G)}:\vspace{-1em}
    \begin{addmargin}[1em]{2em} 
       \STATE {\sc Stack-B}.push$((V(G)$, $\varepsilon))$ \hfill\textit{\%Place $(V(G)$, $ \varepsilon)$ in {\sc Stack-B}}
        
        \WHILE {{\sc Stack-B} is not empty}
            \STATE $(U, index_U) \gets $ {\sc Stack-B}.pop() \hfill \textit{\%Get (U, $index_U$) from top of {\sc Stack-B}}
              \STATE $pre \gets \operatorname{DelLast}(index_U)$ \hfill \textit{\%Remove last character from $index_U$}
            \IF{{\sc Stack-C} is not empty}
                    
            \WHILE{$pre \neq$ to the index of {\sc Stack-C}.peek() }
                \STATE $(S, index_S ) \gets ${\sc Stack-C}.pop()  \hfill  \textit{\%Pop from {\sc Stack-C} }%Remove top element from {\sc Stack-C}
                
                Remove cops from all vertices in $S$  
            \ENDWHILE
        \ENDIF
            \STATE \texttt{Separate}(G[U], $index_U$, |V(G)|)
        \ENDWHILE
\end{addmargin}

\vspace{0.2cm}
\STATE {\textbf{Function} \texttt{Separate($H$, $index$, n)}:}\vspace{-1em}
\begin{addmargin}[1em]{2em}
      \IF {$|V(H)| \le f(n)$} 
      \STATE \texttt{Clear}($H$) \hfill \textit{\%Clear $H$}%Clear($G$) and remove top element from {\sc Stack-B}
      \ELSE 
            \STATE $(A,B,C) \gets \mathcal{S}(H)$         \hfill \textit{\%Call Separation algorithm on $H$}
            \STATE Place cops in $C$ 
            \STATE {\sc Stack-C}.push($(C$, $w)$) \hfill \textit{\%Push $(C$, $w)$ onto {\sc Stack-C}}%Put cops in $C$ and write $(C,w)$ to {\sc Stack-C}
            \STATE {\sc Stack-B}.push($(B, w\cdot 0)$) \hfill \textit{\%Push $(B$, $w\cdot 0)$ onto {\sc Stack-B}}$\!\!$
            \STATE {\sc Stack-B}.push($(A, w\cdot 1)$) \hfill \textit{\%Push $(A, w\cdot 1)$ onto {\sc Stack-B}}
        \ENDIF
\end{addmargin}
\vspace{0.2cm}
{\textbf{Function} \texttt{Clear}($H$):}}\vspace{-1em}
\begin{addmargin}[1em]{2em}
    \STATE Place cops in all vertices of $H$
    \STATE   Remove cops from all vertices of $H$
\end{addmargin}
    \end{algorithmic}
\end{algorithm}

\subsection{Correctness of the STT algorithm}

Let us consider a node $w$ of a binary separation tree $\mathcal{T}_{\mathsf{sep}}$ with respect to the separation algorithm $\mathcal{S}$ applied to a graph $G$ that corresponds to the subset of $V(G)$. Clearly, the node can be denoted by a binary string. Upon applying \( \mathcal{S} \) to \(G[w] \), three  sets are obtained: \( A_w \), \( B_w \), and \( C_w \). In our notation, the set \( A_w \) is denoted as a sting \( w \cdot ``1" \), the set \( B_w \) as a sting \( w \cdot ``0" \), and the set \( C_w \) as \( \operatorname{Sep}(w) \), respectively. The following lemma is key to the analysis of the \texttt{STT} algorithm. Note that we take $x_1$ to be the empty string $\varepsilon$.
\begin{lemma}\label{lem:key}
 Let a robber be on some vertex of a %planar
 graph $G$ and $x_1\cdots x_n \in \{0,1\}^*$, where $n\geq 1$ and $\{0,1\}^*$ is the family of all binary strings.
 If there is a cop at each vertex of $\cup_{i=1}^n \operatorname{Sep}(x_1\cdots x_i)$, then the robber in $x_1\cdots x_n \cdot ``1"$ or $x_1\cdots x_n\cdot ``0"$ cannot leave their respective set without being caught.
 \end{lemma}
 
\begin{proof} 
  We prove the lemma by induction on \(n\).  
  For the base case \(n = 1\), we apply the separation algorithm \(\mathcal{S}\) to \(G\) and obtain three sets \((1, 0, \operatorname{Sep}(\varepsilon))\). Now place one cop on each vertex of the set \(\operatorname{Sep}(\varepsilon)\). Since \(\operatorname{Sep}(\varepsilon)\) separates set \(0\) from \(1\), a robber in either of set \(0\) or \(1\) cannot move to the other set without being caught.

Assume the statement holds for \(n=k\). Thus, w.l.o.g.\ the robber is at a vertex in  \(G[x_1\cdots x_{k+1}]\) and cannot leave this subgraph without being caught by a cop in the set \(\cup_{i=1}^k \operatorname{Sep}(x_1\cdots x_i)\). Applying \(\mathcal{S}\) to  \(G[x_1\cdots x_{k+1}]\) gives three sets: \(x_1\cdots x_{k+1}\cdot``1"\), \(x_1\cdots x_{k+1}\cdot``0"\), and \(\operatorname{Sep}(x_1\cdots x_{k+1})\), where we place cops in the latter. Since the robber cannot leave  \(G[x_1\cdots x_{k+1}]\), if it was not caught by the placement of cops in \(\operatorname{Sep}(x_1\cdots x_{k+1})\), then it must be in one of the two sets \(x_1\cdots x_{k+1}\cdot``1"\), \(x_1\cdots x_{k+1}\cdot``0"\). However, by virtue of $\mathcal{S}$,  it cannot reach the other set without passing though \(\operatorname{Sep}(x_1\cdots x_{k+1})\), where it will be caught.  \qed
\end{proof}

Let the robber's territory in round \( i \geq 0 \) of the game be  the set of vertices in the graph \( G \) where a robber, which has not yet been caught, can possibly be located at the end of the round with respect to the cops' strategy. For example, in round \( i = 0 \), the robber's territory is given by
$ V(G) \setminus \{c_0(0), c_1(0), \dots, c_{k-1}(0)\}$.

\begin{theorem}\label{thm:correctness}
Let \( R_i \) be the territory of the robber after the \( i \)-th iteration of the outer while loop of {\upshape\texttt{Scheduler(G)}}.
Then, \( \forall  k \in \mathbb{N} \)  such that $|R_k|>0$, we have \( R_{k+1} \subsetneq R_k \). Furthermore, the {\upshape\texttt{STT}}~algorithm terminates in a finite number of steps with the robber being caught.
\end{theorem}
\begin{proof}Let $W$ denote the set of vertices contained in some set on {\sc Stack-B}. We will prove the following invariant (I) of the while loop holds by induction: $(I) \text{ At then end of the $i$th loop, $R_i=W$.}$ 

For the base case, $R_0=V(G)$, this holds as before the while loop begins $(V(G), \varepsilon)$ is pushed to the stack and no cops have been placed yet. 

Suppose that invariant (I) holds at the end of the $i$th iteration, and that $R_i\neq \emptyset $. Thus, by induction, {\sc Stack-B} is non-empty and the first element $(U,index_U)$ is popped from {\sc Stack-B}. The next step of the outer while loop, is a check to see if {\sc Stack-C} is empty, followed by an (inner) while loop. This inner while loop pops elements from {\sc Stack-C} until the top one is a separator $S$ corresponding to an output $(U,U',S)$ from $\mathcal{S}$, that is, $S$ separates $U$ from its sibling $U'$ in the separation tree. Note that a separator is only popped from {\sc Stack-C} if the two sets it separates have already been cleared. By \cref{lem:key}, and since all separators above $S$ in the separation tree are still on {\sc Stack-C}, we have that this step does not change $R_i$. The final step in the loop is to  call $ \texttt{Separate}(U, index_U, n)$.  During this call, either the whole set $U$ is cleared with cops (thus $R_{i+1} =R_i\setminus U$) and thus $(I)$ holds since $U$ was the only set of vertices removed from {\sc Stack-B}. Or $f(n)$ cops are placed on a separator of $G[U]$, and the two corresponding separated sets are added to {\sc Stack-B}. Thus also in this case, since the robber cannot be in the separator, invariant $(I)$ holds.  

Observe that the (potentially empty) union of the sets added during one iteration is a strict subset of the set $U$ popped at the start of the iteration, i.e.\ a positive number of vertices is removed in each iteration. Since the graph is finite it follows that the while loop terminates in a finite number of iterations. 

Since invariant $(I)$ holds, at the end of the final iteration the set $R_i=\emptyset $, and so the robber must be caught. Additionally, if $|R_i|>0$, then there will another iteration of the loop and so, by the above paragraph, we have $R_{i+1}\subsetneq R_i$. 

Finally, although the \texttt{STT} algorithm is defined for teleporting cops, using cops which travel over edges (waiting for them to arrive) does not affect correctness -  as the robber can't escape its current territory in the meantime.\qed
\end{proof}

\cref{thm:correctness} shows that for teleporting cops the \texttt{STT} algorithm gives a monotonic strategy in the sense of \cite{dereniowski2015zero}. However, for cops moving over edges the robber's territory may not be monotonically decreasing, as this depends on the paths taken by the individual cops when travelling between separator sets.

\subsection{Capture time, cop number and computational complexity}

In this subsection, we define a recursive function that provides an estimate of the number of vertices in the binary separation tree. The estimate will be used in the next subsection to derive upper bounds on the \texttt{STT} algorithm's computation time and memory usage, upper bounds on the cop number, and capture time associated with the \texttt{STT} algorithm.

For any $f:\mathbb{N}\rightarrow \mathbb{N}$ and $n\in \mathbb{N}_+$, let $\varphi(x):=\varphi_{f(\cdot),n}(x)$ denote the maximum number of nodes in the binary separation tree corresponding to subgraphs of size $x$, assuming we place $f(x)$ cops on each separator, and stop the recursion when we reach a graph of size $f(n)$. This mirrors what happens in Algorithm \ref{alg:STT}.
For simplicity, we assume each separation of a subgraph on $x$ vertices uses exactly $f(x)$ cops, even if the actual separator is smaller. This removes $f(x)$ vertices, leaving $x - f(x)$ vertices in two disconnected parts, each with at most $\frac{2x}{3}$ vertices.
Thus, the following recurrence holds: 
\begin{equation}\label{eq:treesize}\!\!\!\!\!\!\varphi(x) = \begin{cases} 0     &\text{for all }x< 1;\\
1    &\text{for all }1\leq x\leq f(n);\\ \max\limits_{\lceil x/3 \rceil \leq y\leq \lfloor  2x/3\rfloor} \{\varphi(y) + \varphi(x-f(x)-y) \}  + 1& \text{otherwise.} \end{cases} 
\end{equation}   
\noindent
The recurrence reflects two cases: when $x \leq f(n)$, the region can be cleaned directly with one full group of cops, so $\varphi_k(x) \leq  1$. For $x > f(n)$, we place $f(x)$ cops on a separator and split the remaining $x - f(n)$ vertices into two parts. The recurrence adds 1 for the current separator node, and recursively counts the nodes in the two resulting subtrees. The maximisation over $y$ ensures all $(f(x), 2/3)$-separators are considered. 

\newcommand{\septreesize}{ For any $f:\mathbb{N}_+\rightarrow \mathbb{N}_+$ and $n\in \mathbb{N}_+$, let $ \varphi_{f(\cdot),n}(x)$ given by \eqref{eq:treesize}. Then, for all $x\geq  f(n)$, 	\[\varphi_{f(\cdot),n}(x) \leq  \frac{6x}{f(n)} -1 .  \] }
\begin{proposition}\label{prop:treesize} 
\septreesize
\end{proposition}

%It is important to \textit{note that from here on}, unless otherwise stated, \( n \) \textit{denotes the number of vertices in a graph under consideration}.

\begin{corollary}\label{cor:treesize}
Let $\mathcal{C}$ be any $f$-separable hereditary class, and $G \in \mathcal{C}$ have $n$ vertices.
Then, the number of nodes in the binary separation tree \(\mathcal{T}_{\mathsf{sep}}\) is \(O\big( \frac{n}{f(n)}\big) \).
\end{corollary}

We can then use this to bound the time and memory complexity of  $\texttt{STT}$. 
 
\begin{proposition}\label{prop:run_time}
    The {\upshape\texttt{STT}} algorithm runs in time \(O\big(\frac{n}{f(n)} \cdot T(n)\big)\), where %\( n \) is the number of vertices in the input graph and 
    \( T(n) \) is the running time of a separation algorithm.
    
\end{proposition}
\begin{proof}
The computational time of the \texttt{STT} algorithm is dominated by the number of calls within the \texttt{while} loop that repeatedly invokes the separation algorithm \( \mathcal{S} \), which takes as input a set of vertices from a subgraph of graph \(G\).

Observe that \( \mathcal{S} \) is executed once for each internal node of the binary separation tree \(\mathcal{T}_{\mathsf{sep}}\). Consequently, the runtime is dominated by \(T(n)\), the max running time of a call to \( \mathcal{S} \), multiplied by the number of internal vertices of the tree \(\mathcal{T}_{\mathsf{sep}}\), which is \(O\big(\frac{n}{f(n)}\big)\) by \cref{cor:treesize}. Thus, the time complexity is \(O\big(\frac{n}{f(n)} \cdot T(n)\big)\). \qed
\end{proof}

Let us introduce a function $\ell: \mathbb{N}\rightarrow \mathbb{N}$ that will be used to evaluate the cop number and space complexity. The function $\ell: \mathbb{N}\rightarrow \mathbb{N}$ 
is given by 
\begin{equation}\label{eq:2}
\ell(n):= \left\lceil  \log_{3/2} n\right\rceil + 1. \end{equation} 
Note that $\ell(n) \approx 2.466 \cdot \ln n  $, for large $n$. For $f: \mathbb{N}_+\rightarrow \mathbb{N}_+$ define 
\begin{equation}\label{eq:copbound} C_f(n) = f(n)+ \sum_{i=0}^{\ell(n)} f\Big(\Big\lfloor \Big(\frac{2}{3}\Big)^i\cdot n\Big\rfloor \Big).\end{equation}
Thus $C_f(n) = O(f(n) \log n )  $, but the $\log n $ factor may vanish if $f(n)$ grows  fast.

\begin{proposition}\label{prop:memory} 
    The {\upshape\texttt{STT}} algorithm uses $O(n)$ memory in the RAM model for any connected graph $G$ as input to the algorithm.
\end{proposition}
\begin{proof}
 In the \texttt{STT} algorithm, memory is used only to push subgraphs and their indices in {\sc Stack-B} and {\sc Stack-C}. Since {\sc Stack-B} and {\sc Stack-C} contain the same number of sets, we need to only estimate the number and size of sets stored in {\sc Stack-B}.
Additionally, we assume that the indices are negligible compared to the corresponding sets. 
The maximum height $h$ of the binary separation tree is given by $\ell(n) = \lceil \log _{3/2} n +1 \rceil$. 

Since the algorithm stores in {\sc Stack-B} all sets $B$ along the path from the root of the tree to its leaf, and the size of each set at level $i$ is at most $\left(\frac{2}{3}\right)^i n$, we will need 
$
M = \sum_{i=0}^{\lceil \log_{3/2} n\rceil + 1} \left(\frac{2}{3}\right)^i  n   =  O(n) %\approx 3n 
$
memory. \qed 
\end{proof}

So far, we have shown the correctness of the \texttt{STT} algorithm in \cref{thm:correctness}, and we have obtained the results on the time and space complexity in \cref{prop:run_time,prop:memory}. Now, we will analyse the number of cops used by the \texttt{STT} algorithm in order to bound the zero-visibility cop number of a graph $G$ and derive an upper bound on the capture time with respect to the algorithm.

\begin{theorem}
\label{prop:sst_cop_number}
	Let $\mathcal{C}$ be any $f$-separable hereditary class, and $C_f(\cdot)$ be given by \eqref{eq:copbound}. Then, for any $n$-vertex $G\in \mathcal{C} $, the {\upshape\texttt{STT}} algorithm uses $C_f(n)$ cops, and thus $\cop_0(G) \leq C_f(n)$. 
\end{theorem}

\begin{proof}
The number of `teams' of cops  deployed by the \texttt{STT} algorithm on the separation sets, where the size of the $i$-th team is $f(\lfloor(2/3)^{i} n\rfloor)$, is at most the height $h$ of the binary partition tree $\mathcal{T}_{\mathsf{sep}}$. To bound $h$, observe that the output sizes of $A$ and $B$ by the separation algorithm $\mathcal{S}$ used in the \texttt{STT} algorithm are at most $2/3$ times the size of the input. Also, once sets have size at most $f(n)$ then they are just cleared using at most $f(n)$ cops. Since $f(n)\geq 1$, it follows that $(2/3)^{h-1} n \geq 1 $, and so rearranging gives $h\leq \ell(n)$, where $\ell(n)$ is given by \eqref{eq:2}. Summing the cops in all these teams, and not forgetting the last team of size $f(n)$ used to clear the sets at the leaves, gives $C_f(n)$ cops. \qed
\end{proof}

\noindent A function $f:D\to \mathbb{R}$ is multiplicative if $f(xy)= O(f(x)f(y)) $ for all~$x,y\in D.$ 

\newcommand{\multprop}{Let $f(x)\leq x $ be a  non-decreasing multiplicative function, and $\mathcal{C}$ be an $f$-separable hereditary class. Then, for any $n$-vertex $G\in \mathcal{C} $, the {\upshape\texttt{STT}} algorithm uses $O(f(n))$ cops, and thus $\cop_0(G)= O(f(n))$. }
     \begin{proposition}\label{prop:sublinear_mult_func_stt_cop_numb}
    \multprop
\end{proposition}
We also show a novel bound on the zero-visibility capture time. 

\begin{theorem}\label{thm:stt_capt_time}
Let $\mathcal{C}$ be any hereditary class that is $f$-separable, \( G \in \mathcal{C} \) and \( D_G \) be diameter of \( G \). Then, the capture time of \( G \) with respect to the {\upshape\texttt{STT}} algorithm is bounded above by \( O\Big(D_G \cdot \frac{n}{f(n)}\Big) \).
\end{theorem}
\begin{proof}The cop team must visit every vertex of the binary separation tree \(\mathcal{T}_{sep}\) of $G$ to capture the robber. Let \(u, v \) be two vertices of the tree \(\mathcal{T}_{sep}\).  In the worst case, to go from \(u\) to \(v\), steps of \(D_G\) are required. By \(\cref{prop:treesize}\), the number of nodes in $\mathcal{T}_{\mathsf{sep}}$ is bounded above by $\frac{n}{f(n)}$, where \(f(n)\) is the size of the separator. Since the team of cops must traverse each of the vertices of \(\mathcal{T}_{sep}\), we will get the desired estimate \(O(D_G \cdot \frac{n}{f(n)})\). \qed
\end{proof}

We get the following from the Lipton-Tarjan separation algorithm \cite{LiptonT80}. 
\begin{corollary}
 Let \( G  \) be a planar graph with diameter \( D_G \), and $k=C_{\sqrt{\cdot}}(n)=O(\sqrt{n})$ be given by \cref{eq:copbound}. Then, the SST algorithm gives \( \capt_k(G) =O\left(D_G \cdot \sqrt{n}\right) \). 
\end{corollary} 
 
\subsection{Applications of STT algorithm}

In order to demonstrate our bounds on the zero-visibility cop number  and capture time, we apply these result to the hyperbolic random graph. 
In addition, we demonstrate the advantages of the \texttt{STT} algorithm in the context finding the \textit{approximate zero-visibility cop number} $\apx_0(G)$ which is the restriction of $\cop_0(G)$ to stratgies that can be computed in polynomial time.

The hyperbolic random graph exhibits many properties of large real-world networks, such as a power-law degree distribution, a constant clustering coefficient, polylogarithmic diameter, and other relevant structural features \cite{krioukov2010hyperbolic,boguna2010sustaining,blasius2016hyperbolic} - which make it a popular model. Moreover, to the best of our knowledge, there are currently no known exact results on the pathwidth of hyperbolic random graphs, so application of the \texttt{PW} algorithm remains undesirable at this stage.

The hyperbolic random graph has three parameters: the number of vertices \( n \); the parameter \( \alpha \geq \frac{1}{2} \), which controls the power-law exponent; and the parameter \( C \), which controls the average degree. For our purposes, as often in the literature, we just assume that $C$ is some suitably large constant. The hyperbolic random graph \( G \sim \mathcal{G}_{\mathrm{hyp}}(n, \alpha) \) can be obtained by sampling \( n \) points in the disk \( D_R \) of radius \( R = 2 \log n + C \), using radial coordinates \( (r, \theta) \) with the center of \( D_R \) as the origin. The angle \( \theta \) is drawn uniformly from \( [0, 2\pi] \), while the radius \( r \) is chosen according
to the density $
d(r) = \frac{\alpha \sinh(\alpha r)}{\cosh(\alpha R) - 1}.
$

\newcommand{\hrgprop}{
Let $G$ be any largest component of $\mathcal{G}_{\mathrm{hyp}}(n, \alpha)$ given with its geometric representation, where \(\alpha \geq \frac{1}{2}\). Then, w.h.p.\ the run time of the {\upshape\texttt{STT}} algorithm is $O\big(n^{1+ \min\{\alpha,1\}}\big)$, and gives $\cop_0(G)\leq k$ where\\
%bounded  by:\\
\begin{tabular}{@{}ll@{}}
$
k =
\begin{cases}
O\big(n^{1-\alpha}\big) & \text{for } \alpha < 1, \\
O\big(\log^3 n \big) & \text{for } \alpha = 1, \\
O\big(\log^2 n\big) & \text{for } \alpha > 1.
\end{cases}
$ & \hspace{0.2cm} \& \hspace{0.2cm}
$\capt_k(G)=
\begin{cases}
O\big( n^{\alpha}\log n\big) & \text{for } \alpha < 1, \\
O\big(\frac{n}{\log^2 n }\big) & \text{for } \alpha = 1, \\
n^{1+\frac{1}{2\alpha }} \log^{O(1)} n & \text{for } \alpha > 1.
\end{cases}$
\end{tabular}}

\begin{proposition}\label{Cop-hyperbolic-random}
\hrgprop 
\end{proposition}

Since finding an optimal pathwidth decomposition or optimal separation set is NP-complete, we show an improved approximate zero-visibility cop number using the alternative approach based on separation hierarchy.
Applying the available approximation algorithms in the literature for finding separators and path decompositions \cite{feige2005improved},
 we provide a comparison of the approximate zero-visibility cop numbers obtained by the \texttt{STT} and \texttt{PW} algorithms on various classes.

 \newcommand{\approxcop}{Let $\mathcal{C}$ be any $f$-separable hereditary class, $G\in \mathcal{C} $ have $n$ vertices, and $C_f(n)=O(f(n)\log n)$ be given by \eqref{eq:copbound}. Then, the {\upshape\texttt{STT}} algorithm gives  
\[
\apx_0(G) =
\begin{cases}
O\Big(C_f(n) \cdot \sqrt{\log f(n)} \Big) & \text{if }  G \in \mathcal{C}, \\
O\Big(C_f(n)\cdot |V(H)|^2\Big) & \text{if } G \text{ is } H\text{-minor-free,  }, \\
O\big(C_f(n) \cdot g\big) & \text{if } G \text{ has genus at most } g.
\end{cases}
\]}
\begin{proposition}\label{prop:stt_cop_number_approx}
\approxcop
\end{proposition}

Similarly, we can apply the best known approximation  for Pathwidth~\cite{feige2005improved}. 
\newcommand{\pwapprox}{Let $\mathcal{C}$ be any $f$-separable hereditary class, $G\in \mathcal{C} $ have $n$ vertices. Then, the {\upshape\texttt{PW}} algorithm gives 
\[
\apx_0(G) =
\begin{cases}
O\Big(\pw(G) \cdot  \sqrt{\log(\pw(G))}\cdot  \log n \Big) & \text{if }  G \in \mathcal{C}, \\
O\Big(\pw(G) \cdot |V(H)|^2\cdot \log n\Big) & \text{if } G \text{ is } H\text{-minor-free}.\\
\end{cases}
\]
}
\begin{proposition}\label{prop:pw_cop_number_approx}
\pwapprox
\end{proposition}

Generally speaking, $\pw(G)$ is harder to approximate than $\s(G)$, and the best known approximations \cite{feige2005improved} all carry an additional $\log n$ factor over those for $\s(G)$. As a concrete examples of where the \texttt{STT} algorithm outperforms the \texttt{PW} algorithm in its bound of $\apx_0(G)$ we can consider planar graphs. Since the separator number is $O(\sqrt{n})$ and planar graphs are $K_6$-minor-free, we can apply Propositions \ref{prop:pw_cop_number_approx} and \ref{prop:sublinear_mult_func_stt_cop_numb} to give $\apx_0(G) =O(\sqrt{n})$ via the \texttt{STT} algorithm, however the corresponding bound from \texttt{PW} is only $\apx_0(G) =O(\sqrt{n}\log n)$. In particular a constant factor approximation to $\pw(G)$ is only known for outerplanar graphs \cite{govindan1998approximating}. A similar gap holds for $K_t$-minor-free graphs, where $t\geq 6$. 

\medskip 

\noindent {\bf Conclusion.} As the future direction the reduction of space complexity of the \texttt{STT} algorithm could be analysed further, e.g.~by storing  only the order of separations  and potentially using the LOGSPACE algorithm for recursively computing separation sets on the way, following \cite{logspace-separators}, with careful application and analysis.

\noindent {\bf Acknowledgements.} 
Tymofii Prokopenko was supported by the Centre for
Doctoral Training in Distributed Algorithms. 
\bibliographystyle{plain}
\bibliography{bibliography}

\appendix  
\section{Appendix}\label{Appendix-B}

Since some of our results rely on the Lipton-Tarjan theorem for partitioning planar graphs with $n$ vertices (with the additional requirement that the cost of each vertex is $\frac{1}{n}$) into disjoint parts, we include the theorem for completeness.

\begin{theorem}[Lipton-Tarjan, \cite{LiptonT80}]\label{thm:LT}
Let $G$ be any $n$-vertex planar graph with non-negative vertex cost summing to not more than one. Then the vertices of $G$ can be partitioned into three sets $A,B,S$ such that no edge joins a vertex in $A$ with the vertex in $B$, neither $A$ nor $B$ has total vertex cost exceeding $\frac{2}{3}$, and $S$ contains no more than $2\sqrt{2}\sqrt{n}$ vertices. Furthermore, there exists an algorithm for which $A,B,S$ can be found in $O(n)$ time.
\end{theorem}

\noindent
{\bf Proposition \ref{prop:treesize}.}
{\it \septreesize }
\begin{proof}
	For ease of notation, as $n$ and $f(\cdot)$ are fixed, we denote $\varphi(x):=\varphi_{f(\cdot),n}(x)$. We will prove the proposition by induction on $x$. For the base case we have $x = f(n)$. Thus, by the boundary condition from \eqref{eq:treesize}, we have  $\varphi(x) = 1 \leq   \frac{6f(n)}{f(n)} -1$.
	
	Suppose that the proposition holds for any $f(n)\leq y \leq x$, that is, $\varphi(y) \leq  \frac{6x}{f(n)} -1$. Then, we need to show that $\varphi(x+1) \leq \frac{6(x+1)}{f(n)} -1 $.

	We can assume w.l.o.g. that $x+1> f(n)$ as we have established the base case. Thus, by \eqref{eq:treesize}, %of $\varphi_k(x)$
	we have
	\begin{equation}\label{eq:cases} \varphi(x+1) = \max_{\lceil (x+1)/3 \rceil \leq y \leq \lfloor 2(x+1)/3 \rfloor} \{ \varphi(y) + \varphi(x+1-f(x)-y) \} + 1. \end{equation}
	There is the awkward fact that one or both of the arguments of the terms in the max above may be strictly below $f(n)$. In this case we cannot apply the inductive hypothesis to them. However, by the definition \eqref{eq:treesize} of $\varphi_{f(\cdot),n}(\cdot)$, we know that in this case the offending term is equal to $1$. We proceed by cases. 
	
	If both terms have argument less than $f(n)$, for example if $ \lfloor 2(x+1)/3 \rfloor \leq f(n) $, then by \eqref{eq:cases} we have 
	\[\varphi(x+1) = 2 + 1 \leq \frac{6(x+1)}{f(n)} -1 .\]   
If exactly one of the terms has argument less than $f(n)$, then 
	\begin{align*} \varphi(x+1) &\leq \max_{x+1-f(x)-\lfloor 2(x+1)/3 \rfloor \leq y \leq \lfloor 2(x+1)/3 \rfloor}   \varphi(y)     + 2  \\
		 &\leq \max_{x+1-f(x)-\lfloor 2(x+1)/3 \rfloor \leq y \leq \lfloor 2(x+1)/3 \rfloor}  \frac{6y}{f(n)}  -1    + 2  \\
		&= \frac{6 \lfloor 2(x+1)/3 \rfloor}{f(n)}  +1  \\
		&\leq  \frac{6(x+1)}{f(n)} -1 , \end{align*} 
where we applied the induction hypothesis to the other term. Otherwise,
	\begin{align*}
		\varphi(x+1) &\leq  \max_{\lceil (x+1)/3 \rceil \leq y \leq \lfloor 2(x+1)/3 \rfloor} \left\{ \frac{6y}{f(n)} -1 + \frac{6(x-f(x) - y)}{f(n)} -1 \right\} + 1\\
		&= \frac{6(x-f(x))}{f(n)} -1\\
		&\leq  \frac{6(x+1)}{f(n)} -1.
	\end{align*}
Thus, in all three cases the inductive step holds. \qed 
\end{proof}

\noindent
{\bf Proposition~\ref{prop:sublinear_mult_func_stt_cop_numb}.}
{\it  \multprop }

\begin{proof} 
By the conditions on $f$, for all $i\geq 0$, we have \[f\Big(\Big\lfloor \Big(\frac{2}{3}\Big)^i \cdot n \Big\rfloor\Big) \leq f\Big(\Big(\frac{2}{3}\Big)^i \cdot n \Big)  =O\Big(f(n) \cdot f\Big(\Big(\frac{2}{3}\Big)^i \Big) \Big) = O(f(n))\cdot \Big(\frac{2}{3}\Big)^i.\] Thus applying \cref{prop:sst_cop_number} in this case we obtain
\[
\cop_0(G) \leq C_f(n)=f(n)+\sum_{i=0}^{\ell(n)} f\Big(\Big\lfloor \Big(\frac{2}{3}\Big)^i \cdot n \Big\rfloor\Big) = O(f(n)) \cdot \sum_{i=0}^{\ell(n)}\Big(\frac{2}{3}\Big)^i = O(f(n)), 
\]
 where we bounded the last sum by a geometric series.  \qed\end{proof}

\noindent
{\bf Proposition~\ref{Cop-hyperbolic-random}.}
{\it  \hrgprop }

\begin{proof}To prove this result we apply our results using several previously established properties of hyperbolic random graphs. 

Firstly, by the bounds on treewidth from \cite[Theorem 9]{blasius2016hyperbolic} and as $ \s(G) =\Theta(\trw(G))$ \cite[Theorem 1]{DvorakN19}, if $G$ is any largest component of a hyperbolic random graph $\mathcal{G}_{\mathrm{hyp}}(n, \alpha)$, then with high probability $\s(G) \leq f(n)$ where 
\begin{equation}\label{eq:fhyp}f(n) = \begin{cases}
O\left(n^{1-\alpha}\right) & \text{for } \alpha < 1, \\
O(\log^2 n) & \text{for } \alpha = 1, \\
O(\log n) & \text{for } \alpha > 1.
\end{cases}\end{equation}
 Recall that by \cref{prop:sst_cop_number} we have $\cop_0(G) \leq C_f(n) $ where, by \eqref{eq:2} and \eqref{eq:copbound},  
\begin{equation}\label{eq:Cfrecap}
  C_f(n) = f(n)+ \sum_{i=0}^{\ell(n)} f\Big(\Big\lfloor \Big(\frac{2}{3}\Big)^i\cdot n\Big\rfloor \Big), \quad \text{and} \quad  \ell(n):= \left\lceil  \log_{3/2} n\right\rceil + 1. \end{equation} 
Thus by inserting the bounds from \eqref{eq:fhyp} in to \eqref{eq:Cfrecap} we have 
\begin{equation}\label{eq:hypbounds}
C_f(n) =
\begin{cases}
\sum_{i=0}^{\left\lceil  \log_{3/2} n\right\rceil + 1} O\left(((2/3)^i n)^{1-\alpha}\right) =  O\left(n^{1-\alpha}\right) & \text{for } \alpha < 1, \\
 \sum_{i=0}^{\left\lceil  \log_{3/2} n\right\rceil + 1} O(\log^2 n)  = O(\log^3 n) & \text{for } \alpha = 1, \\
\sum_{i=0}^{\left\lceil  \log_{3/2} n\right\rceil + 1} O(\log n) = O(\log^2 n ) & \text{for } \alpha > 1,
\end{cases}
\end{equation}
where we take $k=C_f(n)$ in the statement of the theorem. 

For the capture time, recall that for this $k$ we have $\capt_k(G) \leq O(D_G \cdot \frac{n}{f(n)} )$ by \cref{thm:stt_capt_time}. Note that, assuming the average degree parameter $\nu$ is suitably large (a common assumption for the RHG) and $1/2<\alpha \leq 1$ then w.h.p.\  $D_G = O(\log n)$. We will also use the fact that for $\alpha >1$ w.h.p.\  any connected component of the hyperbolic random graph has at most $n^{\frac{1}{2\alpha}}\log^{O(1)} n$ vertices \cite[Theorem 1.1]{LargestSub}, thus this bound also holds for the diameter $D_G$. Plugging these bounds on $D_G$, along with those on $f(n)$ from \eqref{eq:hypbounds}, into the bound on $\capt_k(G)$ from \cref{thm:stt_capt_time} gives that w.h.p.\ the claimed bound on the capture times holds.   

We finally address the runtime of the \texttt{STT} algorithm  on the hyperbolic random graph. By \cref{prop:run_time} the \texttt{STT} algorithm runs in time \(O\Big(\frac{n}{f(n)} \cdot T(n)\Big)\), where $T(n)$ is the run time of the separation algorithm. Observe that the separator decomposition in \cite{blasius2016hyperbolic} is based around a fixed partitioning of the hyperbolic disc. It follows that if we are supplied with the hyperbolic representation of the graph $G$, then a separator can be computed in time $O(n)$, that is $T(n)=O(n)$. This, plugging this bound along with the w.h.p.\ bounds on $f(n)$ from \eqref{eq:hypbounds} into \cref{prop:run_time} gives the claimed bound on running time.\qed  
\end{proof}

\noindent
{\bf Proposition~\ref{prop:stt_cop_number_approx}.}
{\it \approxcop
}

\begin{proof}
We will obtain our bounds on the approximate zero-visability cop number by applying the best known approximation ratio for separators from \cite{feige2005improved}.

In particular, suppose we have a polynomial time approximation algorithm for finding separators in an $f$-separable class containing $G$, that has approximation ratio $A(n)$ (where assume w.lo.g.\ that $A(n)$ is non-decreasing). Then, using this approximation ratio together with \cref{prop:sst_cop_number} gives 
\begin{align}\apx_0(G) &= f(n)\cdot A(n) +  \sum_{i=0}^{\ell(n)} f\Big(\Big\lfloor \Big(\frac{2}{3}\Big)^i \cdot n \Big\rfloor \Big)  \cdot A\Big(\Big\lfloor \Big(\frac{2}{3}\Big)^i \cdot n \Big\rfloor \Big) \notag \\
&\leq  A(n) \cdot \Bigg( f(n) +  \sum_{i=0}^{\ell(n)} f\Big(\Big\lfloor \Big(\frac{2}{3}\Big)^i \cdot n \Big\rfloor \Big)  \Bigg)\notag \\
&= A(n) \cdot C_f(n) \label{eq:basicapprox}.  \end{align}

Firstly, by \cite[Theorem 3.11]{feige2005improved}, for any $f$-separable graph $G$ there is a polynomial time approximation algorithm for finding a separating set with an approximation ratio of  \( A(n)= O\big(\sqrt{\log( f(n)) }\big) \).  Thus, for  any $f$-separable $G$,  inserting this bound in \eqref{eq:basicapprox} gives \begin{equation}
    \label{eq:generalbasic}\apx_0(G)= O\Big(C_f(n) \cdot \sqrt{\log f(n)} \Big) .\end{equation}

For the special case of \textit{H}-minor-free graphs, for any graph $H$, there is a $O(|V(H)|^2)$-approximation algorithm for separators by \cite[Theoerm 4.2]{feige2005improved}. Using this in \eqref{eq:basicapprox} yields \( \apx_0(G)  = O(|V(H)|^2\cdot C_f(n)) \) in this case. 

Similarly there is an $O(g)$-approximation algorithm for separtors in graphs of genus at most $g$ by \cite{feige2005improved}. Applying this to \eqref{eq:basicapprox}  gives \( \apx_0(G)  = O(C_f(n)\cdot g) \). 
 \qed
\end{proof}

\noindent
{\bf Proposition~\ref{prop:pw_cop_number_approx}.}
{\it  \pwapprox } 
\begin{proof}
Since we need to move cops consecutively from bag \(B_i\) to \(B_{i+1}\) for all \(i\in {1,...,n-1}\) and cover each vertex of \(G[B_j]\) for any \(B_j\), we have to use $\pw(G)$ cops to `clear' the graph with respect to the \texttt{PW} algorithm.
Now, we will use an approximate path decomposition of \( G \) with width \( \pw_{\text{app}}(G) \), instead of computing an optimal path decomposition. To derive the bound stated in the proposition, we use the best known polynomial time approximation algorithm for $\pw(G)$ by \cite[Corollary 6.5]{feige2005improved}, which has approximation ratio $O(\sqrt{\log \pw(G)}\log n )$. This gives a bound of $\apx_0(G)=O(\pw(G)\sqrt{\log \pw(G)}\log n)  $ from the $PW$ algorithm.  For $H$-minor-free graphs, \cite[Corollary 6.5]{feige2005improved} gives a better approximation ratio of $O(|V(H)|^2\log n)$, which we apply similarly. \qed 
\end{proof}

\end{document}